\documentclass[runningheads]{llncs}
\usepackage{amssymb}
\usepackage{textcomp}
\usepackage{amsmath}
\usepackage{mathtools}

\usepackage{amsthm}
\usepackage{algorithm}
\usepackage[noend]{algpseudocode}
\usepackage{comment}
\usepackage{graphicx}
\usepackage{color}
\usepackage{tikz}
\usepackage{cleveref}
\usepackage{tikz}
\usepackage{quantikz}

\newcommand{\cent}[0]{\mbox{\textcent}}
\DeclareMathOperator{\mean}{\mathbf{E}}
\DeclareMathOperator{\Prob}{\mathbf{P}}
\begin{document}

\title{Shallow Implementation of Quantum Fingerprinting with Application to Quantum Finite Automata\thanks{
The research in Section 4 is supported by Russian Science Foundation Grant 24-21-00406, https://rscf.ru/en/project/24-21-00406/.
The study was funded by the subsidy allocated to Kazan Federal University for the state assignment in the sphere of scientific activities (Project No. FZSM-2024-0013).}
}
%
%
\author{Mansur Ziiatdinov\,$^{1,*}$, Aliya Khadieva\,$^{1,2}$, and Kamil Khadiev\,$^{1}$}
\authorrunning{M. Ziiatdinov, et.al.}
%
\institute{$^1$Institute of Computational Mathematics and Information Technologies, Kazan Federal University, Kazan, Russia \\ $^2$ Faculty of Computing, University of Latvia, R\={\i}ga, Latvia}
\maketitle              
\begin{abstract}
Quantum fingerprinting is a technique that maps classical input word to a quantum state. The obtained quantum state is much shorter than the original word, and its processing uses less resources, making it useful in quantum algorithms, communication, and cryptography. One of the examples of quantum fingerprinting is quantum automata algorithm for \(MOD_{p}=\{a^{i\cdot p} \mid i \geq 0\}\) languages, where $p$ is a prime number.

  However, implementing such an automaton on the current quantum hardware is not efficient.
  Quantum fingerprinting maps a word \(x \in \{0,1\}^{n}\) of length \(n\) to a state \(\ket{\psi(x)}\) of \(O(\log n)\) qubits, and uses \(O(n)\) unitary operations. Computing quantum fingerprint using all available qubits of the current quantum computers is infeasible due to a large number of quantum operations.

  To make quantum fingerprinting practical, we should optimize the circuit for depth instead of width in contrast to the previous works. We propose explicit methods of quantum fingerprinting based on tools from additive combinatorics, such as generalized arithmetic progressions (GAPs), and prove that these methods provide circuit depth comparable to a probabilistic method. We also compare our method to prior work on explicit quantum fingerprinting methods.
\keywords{quantum finite automata, quantum fingerprinting, quantum circuit, shallow quantum circuit, quantum hash}
\end{abstract}

\section{Introduction}

A quantum finite state automaton (QFA) is a generalization of classical finite automaton \cite{SY14,AY21}. Here we use the known simplest QFA model \cite{Moore2000}. Formally, a QFA is 5-tuple \( M = (Q,\allowbreak A \cup \{\cent,\mathdollar \},\allowbreak \ket{\psi_{0}},\allowbreak \mathcal{U},\allowbreak \mathcal{H}_{acc})\), where \(Q = \{q_{1}, \ldots, q_{D}\}\) is a finite set of states, \(A\) is the finite input alphabet,  $\cent,\mathdollar$ are the left and right end-markers, respectively. The state of $M$ is represented as a vector \(\ket{\psi} \in \mathcal{H}\), where $\mathcal{H}$ is the $D$-dimensional Hilbert space spanned by $\{ \ket{q_1},\ldots,\ket{q_D} \}$ (here $\ket{q_j}$ is a zero column vector except its $j$-th entry that is 1). The automaton $M$ starts in the initial state \(\ket{\psi_{0}} \in \mathcal{H}\), and makes transitions according to the operators \(\mathcal{U} = \{U_{a} \mid a \in A\}\) of unitary matrices. After reading the whole input word, the final state is observed with respect to the accepting subspace \(\mathcal{H}_{acc} \subseteq \mathcal{H}\).

Quantum fingerprinting provides a method of constructing automata for certain problems. It maps an input word \(w \in \{0,1\}^{n}\) to much shorter quantum state, its fingerprint \(\ket{\psi(w)} = U_{w}\ket{0^{m}}\), where $U_w$ is the single transition matrix representing the multiplication of all transition matrices while reading $w$  and $\ket{0^m} = \underbrace{\ket{0}\otimes \cdots \otimes \ket{0}}_{m~\text{times}} $. Quantum fingerprint captures essential properties of the input word that can be useful for computation.


One example of quantum fingerprinting applications is the QFA algorithms for \(MOD_{p}\) language~\cite{Ambainis2009}. For a given prime number \(p\), the language \(MOD_{p}\) is defined as  \(MOD_{p} = \{ a^{i} \mid i \text{ is divisible by } p\}\). Let us briefly describe the construction of the QFA algorithms for \(MOD_{p}\).

We start with a 2-state QFA $M_k$, where \(k \in \{1,\ldots,p-1\}\). The automaton \(M_{k}\) has two base states \(Q = \{q_{0}, q_{1}\}\), it starts in the state \(\ket{\psi_{0}} = \ket{q_{0}}\), and it has the accepting subspace spanned by \(\ket{q_{0}}\). At each step (for each letter) we perform the rotation
\[
  U_{a} =
  \begin{pmatrix}
    \cos \dfrac{2\pi k}{p} & \sin \dfrac{2\pi k}{p}\\ \\
    -\sin \dfrac{2\pi k}{p} & \cos \dfrac{2\pi k}{p}\\
  \end{pmatrix}.
\]
It is easy to see that this automaton gives the correct answer with probability 1 if \(w \in MOD_p\). However, if $w \notin MOD_p  $, the probability of correct answer can be close to $0$ rather than $1$ (i.e., bounded below by $ 1- \cos^2(\pi/p)  $). To boost the success probability we use \(d\) copies of this automaton, namely \(M_{k_{1}}\), \ldots, \(M_{k_{d}}\), as described below.

The QFA \(M\) for \(MOD_{p}\) has \(2d\) states: \(Q = \{q_{1,0}, q_{1,1}, \ldots, q_{d,0}, d_{d,1}\}\), and it starts in the state \(\ket{\psi_{0}} = \frac{1}{\sqrt{d}} \sum_{i=1}^{d} \ket{q_{i,0}}\). In each step, it applies the transformation defined as:
\begin{align}
  \ket{q_{i,0}} &\mapsto \cos \frac{2\pi k_{i}}{p} \ket{q_{i,0}} + \sin \frac{2\pi k_{i}}{p} \ket{q_{i,1}} \label{eq:u-transform:0}\\
  \ket{q_{i,1}} &\mapsto -\sin \frac{2\pi k_{i}}{p} \ket{q_{i,0}} + \cos \frac{2\pi k_{i}}{p} \ket{q_{i,1}}\label{eq:u-transform:1}
\end{align}
Indeed, $M$ enters into equal superposition of $d$ sub-QFAs, and each sub-QFA applies its rotation.
Thus, quantum fingerprinting technique associates the input word \(w = a^{j}\) with its fingerprint
\[
    \ket{\psi} = \frac{1}{\sqrt{d}} \sum_{i=1}^{d} \cos \frac{2\pi k_{i}j}{p} \ket{q_{i,0}} + \sin \frac{2\pi k_{i}j}{p} \ket{q_{i,1}}.
\]

Ambainis and Nahimovs~\cite{Ambainis2009} proved that this QFA accepts the language \(MOD_{p}\) with error probability that depends on the choice of the coefficients \(k_{i}\)'s. They also showed that for \(d = 2 \log(2p) / \varepsilon\) there is at least one choice of coefficients \(k_{i}\)'s such that error probability is less than \(\varepsilon\). The proof uses a probabilistic method, so these coefficients are not explicit. They also suggest two explicit sequences of coefficients: cyclic sequence \(k_{i} = g^{i} \pmod p\) for primitive root \(g\) modulo \(p\) and more complex AIKPS sequences based on the results of Ajtai et al.~\cite{Ajtai1990}.


Quantum fingerprinting is versatile and has several applications. Buhrman et al. in \cite{Buhrman2001} provided an explicit definition of quantum fingerprinting for constructing an efficient quantum communication protocol for equality checking. The technique was applied to branching programs by Ablayev and Vasiliev \cite{av2009,av2011,av2013}. Based on this research, they developed the concept of cryptographic quantum hashing \cite{av2013hash,av2014,aav2014,aa2015,aav2016,Ablayev2016a,v2016binary,vlz2017,aakv2018,aav2018,vvl2019,aav2020}. Then different versions of hash functions were applied \cite{v2016,av2020,z2016,z2016group,z2023}. The technique was extended to qudits \cite{av2022}.
This approach has been widely used in various areas such as stream processing algorithms \cite{l2009,l2006}, query model algorithms \cite{aaksv2022}, online algorithms \cite{kk2019disj,kk2022}, branching programs \cite{kk2017,kkk2022,agky16}, development of quantum devices \cite{v2016model}, automata (discussed earlier, \cite{GY17,YS10}) and others.

At the same time, the technique is not practical for the currently available real quantum computers. The main obstacle is that quantum fingerprinting uses an exponential (in the number \(m\) of qubits) circuit depth (e.g., see \cite{KZ22,bsocy2021,salehi2021cost,zkk2023} for some implementations of the aforementioned automaton \(M\)). Therefore, the required quantum volume\footnote{Quantum volume is an exponent of the maximal square circuit size that can be implemented on the quantum computer~\cite{Cross2019QuantumVolume,Wack2021QuantumPerformance}.} \(V_{Q}\) is roughly \(2^{|w| \cdot 2^{m}}\). For example, IBM reports~\cite{IBMEagle} that its Falcon r5 quantum computer has 27 qubits with a quantum volume of 128. It means that we can use only 7 of 27 qubits for the fingerprint technique.


In this paper, we investigate how to obtain better circuit depth by optimizing the coefficients used by $M$: $k_1,\ldots,k_d$.
We use generalized arithmetic progressions for generating a set of coefficients and show that such sets have a circuit depth comparable to the set obtained by the probabilistic method.

We summarize the previous and our results in the next list.
\begin{itemize}
    \item The cyclic method, for some constant \(c > 0\):
    \begin{itemize}
        \item the width is \(p^{c/\log \log p}\)
        \item the depth is \(p^{c/\log \log p}\) 
        \item explored in \cite{Ambainis2009}.
    \end{itemize}
      \item The AIKPS method:
    \begin{itemize}
        \item the width is \(\log^{2+3\epsilon} p\)
        \item the depth is \((1 + 2\epsilon) \log^{1+\epsilon}p\; \log\log p\)
        \item explored in  \cite{Razborov1993}.
    \end{itemize}
    \item The probabilistic method:
    \begin{itemize}
        \item the width is \(4\log (2p) / \varepsilon\)
        \item the depth is \(2\log (2p) / \varepsilon\)
        \item explored in \cite{Ambainis2009}.
    \end{itemize}
     \item The GAPs method (this paper):
    \begin{itemize}
        \item the width is  \(p / \varepsilon^2\) 
        \item the depth is \(\lceil \log p - 2 \log \varepsilon \rceil + 2\)
        \item developed in this paper.
    \end{itemize}
\end{itemize}

Note that \(p\) is exponential in the number of qubits \(m\). The depth of the circuits is discussed in Section~\ref{sec:shallow}.

The paper is an extended version of the \cite{ziiatdinov2023gaps} conference paper that was presented at the AFL2023 conference.

In addition, we perform computational experiments for computing parameters $K$ that minimize the error probability for our circuit and the standard circuit. We show that in the case of an ``ideal'' non-noisy quantum device, the error probability for the proposed circuit is at most twice as large as for the standard circuit. At the same time, in the case of noisy quantum devices (that are current and near-future ones), the error probability is much less than for the standard circuit, and allow us to implement QFA for $MOD_{17}$ language using $4$ qubits.

The rest of the paper is organized as follows. In Section~\ref{sec:preliminaries} we give the necessary definitions and results on quantum computation and additive combinatorics to follow the rest of the paper. Section~\ref{sec:shallow} contains the construction of the shallow fingerprinting function and the proof of its correctness. Then, we present several numerical simulations in  Section~\ref{sec:numeric}.
We conclude the paper with Section~\ref{sec:conclusions} by presenting some open questions and discussions for further research.

\section{Preliminaries}\label{sec:preliminaries}

Let us denote by \(\mathcal{H}^{2}\) two-dimensional Hilbert space, and by \((\mathcal{H}^{2})^{\otimes m}\) \(2^{m}\)-dimensional Hilbert space (i.e.,\ the space of \(m\) qubits). We use bra- and ket-notations for vectors in Hilbert space. For any natural number \(N\), we use \(\mathbb{Z}_{N}\) to denote the cyclic group of order \(N\).

Let us describe in detail how the automaton \(M\) works.
As we outlined in the introduction, the automaton \(M\) has \(2d\) states: \(Q = \{q_{1,0}, q_{1,1}, \ldots, q_{d,0}, d_{d,1}\}\), and it starts in the state \(\ket{\psi_{0}} = \frac{1}{\sqrt{d}} \sum_{i=1}^{d} \ket{q_{i,0}}\). After reading a symbol \(a\), it applies the transformation \(U_{a}\) defined by \eqref{eq:u-transform:0}, \eqref{eq:u-transform:1}:
\begin{align*}
  \ket{q_{i,0}} &\mapsto \cos \frac{2\pi k_{i}}{p} \ket{q_{i,0}} + \sin \frac{2\pi k_{i}}{p} \ket{q_{i,1}} \\
  \ket{q_{i,1}} &\mapsto -\sin \frac{2\pi k_{i}}{p} \ket{q_{i,0}} + \cos \frac{2\pi k_{i}}{p} \ket{q_{i,1}}
\end{align*}
After reading the right endmarker \(\$\), it applies the transformation \(U_{\$}\) defined in such way that \(U_{\$} \ket{\psi_{0}} = \ket{q_{1,0}}\). The automaton measures the final state and accepts the word if the result is \(q_{1,0}\).

So, the quantum state after reading the input word \(w = a^{j}\) is \[\ket{\psi} = \frac{1}{\sqrt{d}} \sum_{i=1}^{d} \cos \frac{2\pi k_{i}j}{p} \ket{q_{i,0}} + \sin \frac{2\pi k_{i}j}{p} \ket{q_{i,1}}.\] If \(j \equiv 0 \pmod p\), then \(\ket{\psi} = \ket{\psi_{0}}\), and \(U_{\$}\) transforms it into accepting state \(\ket{q_{1,0}}\), therefore, in this case, the automaton always accepts. If the input word \(w \notin MOD_{p}\), then the quantum state after reading the right endmarker \(\$\) is
\[
  \ket{\psi'} = \frac{1}{d} \Big( \sum_{i=1}^{d} \cos\frac{2\pi k_{i}j}{p} \Big) \ket{q_{1,0}} + \ldots,
\]
and the error probability is
\[
  P_{e} = \frac{1}{d^{2}} \Big( \sum_{i=1}^{d} \cos \frac{2\pi k_{i}x}{p} \Big)^{2}.
\]

In the rest of the paper, we denote by \(m\) the number of qubits in the quantum fingerprint, by \(d = 2^{m}\) the number of parameters in the set \(K\), by \(p\) the size of domain of the quantum fingerprinting function, and by \(U_{a}(K)\) the transformation defined above, which depends on the set \(K\).

Let us also define a function \(\varepsilon: \mathbb{Z}_{p}^{d} \to \mathbb{R}\) as follows:
\[
  \varepsilon(K) = \max_{x \in \mathbb{Z}_{p}} \bigg( \frac{1}{d^{2}} \Big| \sum_{j=1}^{d} \exp \frac{2\pi i k_{j}x}{p} \Big|^{2} \bigg).
\]
Note that \(P_{e} \le \varepsilon(K)\).

We also use some tools from additive combinatorics. We refer the reader to the textbook by Tao and Vu~\cite{TaoVu2006} for a deeper introduction to additive combinatorics.

An additive set \(A \subseteq Z\) is a finite non-empty subset of \(Z\), an abelian group with group operation \(+\). We refer \(Z\) as the ambient group.

If \(A,B\) are additive sets in \(Z\), we define the sum set \(A + B = \{ a + b \mid a \in A, b \in B\}\). We define additive energy \(E(A,B)\) between \(A,B\) to be
\[
E(A,B) = \bigg| \big\{ (a,b,a',b') \in A \times B \times A \times B \mid a + b = a'+b' \big\} \bigg| .
\]

Let us denote by \(e(\theta) = e^{2\pi i \theta}\), and by \(\xi \cdot x = \xi x / p\) bilinear form from \(\mathbb{Z}_p \times \mathbb{Z}_p\) into \(\mathbb{R} / \mathbb{Z}\). Fourier transform of \(f : \mathbb{Z}_{p} \to \mathbb{Z}_{p}\) is \(\hat{f}(\xi) = \mean_{x \in Z} f(x) \overline{e(\xi \cdot x)}\).

We also denote the characteristic function of the set \(A\) as \(1_{A}\), and we define \(\Prob_Z(A) = \widehat{1_{A}}(0) = |A| / |Z|\).

\begin{definition}[\cite{TaoVu2006}]
  Let \(Z\) be a finite additive group. If \(A \subseteq Z\), we define Fourier bias \(\|A\|_{\mathcal{U}}\) of the set \(A\) to be
  \[
    \| A \|_{\mathcal{U}} = \sup_{\xi \in Z \backslash \{0\}} | \widehat{1_A}(\xi) |
  \]
\end{definition}

There is a connection between the Fourier bias and the additive energy.
\begin{theorem}[\cite{TaoVu2006}]\label{thm:fourier-bias-bound}
  Let \(A\) be an additive set in a finite additive group \(Z\). Then
  \[
    \| A \|^4_{\mathcal{U}} \le \frac{1}{|Z|^3} E(A,A) - \Prob_Z(A)^4 \le \| A \|^2_{\mathcal{U}} \Prob_Z(A)
  \]
\end{theorem}

\begin{definition}[\cite{TaoVu2006}]
  Generalized arithmetic progression (GAP) of dimension \(d\) is a set
  \[
    A = \{ x_0 + n_1 x_1 + \ldots + n_d x_d \mid 0 \le n_1 \le N_1, \cdots, 0 \le n_d \le N_d \},
  \]
  where \(x_0, x_1, \ldots, x_d, N_1, \ldots, N_d \in Z\).
  The size of GAP is a product \(N_1 \cdots N_d\).
  If the size of set $A$, \(|A|\), equals to \(N_1 \cdots N_d\), we say that GAP is proper.
\end{definition}


%

\section{Shallow Fingerprinting}\label{sec:shallow}

Quantum fingerprint can be computed by the quantum circuit given in Figure~\ref{fig:deep-circuit}. The last qubit is rotated by a different angle \(2\pi k_{j} x / q\) in different subspaces enumerated by \(\ket{j}\). Therefore, the circuit depth is \(|K| = t = 2^{m}\). As the set \(K\) is random, it is unlikely that the depth can be less than \(|K|\).

Let us note that fingerprinting is similar to quantum Fourier transform. Quantum Fourier transform computes the following transformation:
\begin{equation}\label{eq:qft}
  \ket{x} \mapsto \frac{1}{N} \sum_{k=0}^{N-1} \omega_{N}^{xk}\ket{k},
\end{equation}
where \(\omega_{N} = e(1/N)\). Here is the quantum fingerprinting transform:
\[
  \ket{x} \mapsto \frac{1}{t} \sum_{j=1}^{t} \omega_{N}^{k_{j}x}\ket{k}.
\]

The depth of the circuit that computes quantum Fourier transform is \(O((\log N)^{2})\), and it heavily relies on the fact that in Eq.~\eqref{eq:qft} the sum runs over all \(k = 0, \ldots, N-1\). Therefore, to construct a shallow fingerprinting circuit we desire to find a set \(K\) with special structure.


Suppose that we construct a coefficient set \(K \subset \mathbb{Z}_{p}\) in the following way. We start with a set \(T = \{t_{1}, \ldots, t_{m}\}\) and construct the set of coefficients as a set of sums of all possible subsets:
\[
  K = \Big\{ \sum_{t \in S} t \mid S \subseteq T \Big\},
\]
where we sum modulo \(p\).


The quantum fingerprinting function with these coefficients can be computed by a circuit of depth \(O(m)\)~\cite{Kalis2018} (see Figure~\ref{fig:shallow-circuit}).

Finally, let us prove why the construction of the set \(K \subset \mathbb{Z}_{p}\) works.
\begin{theorem}
  Let \(\varepsilon > 0\), let \(m = \lceil \log p - 2 \log \varepsilon \rceil\) and \(d = 2^{m}\).

  Suppose that the number \(t_{0} \in \mathbb{Z}_{p}\) and the set \(T = \{t_{1}, \ldots, t_{m}\} \subset \mathbb{Z}_{p}\) are such that
  \[
    B = \{ 2 t_{0} + n_1 t_1 + \cdots + n_m t_m \mid 0 \le n_1 < 3, \ldots, 0 \le n_m < 3\}
  \]
  is a proper GAP.

  Then the set \(A\) defined as
  \[
    A = \left\{ t_{0} + \sum_{t \in S} t \mid S \subseteq T \right\}
  \]
  has \(\varepsilon(A) \le \varepsilon\).
\end{theorem}

Let us outline the proof of this theorem. Firstly, we estimate the number of solutions to \(a + b = n\). Secondly, we use it to bound the additive energy \(E(A,A)\) of the set \(A\). Thirdly, we bound the Fourier bias \(\|A\|_{{\mathcal{U}}}\). Finally, we get a bound on \(\varepsilon(A)\) in terms of \(p\) and \(m\).

\begin{proof}
  Let us denote a set \(R_{n}(A)\) of solutions to \(a + b = n\), where \(a,b \in A\) and \(n \in \mathbb{Z}_{p}\):
  \[
    R_{n}(A) = \{ (a,b) \mid a+b = n; \; a,b \in A\}.
  \]
  Note that we have \(E(A,A) = \sum_{n\in Z} {R_n(A)}^2\).

  Suppose that \(n\) is represented as \(n = 2t_{0} + \sum_{i=1}^m \gamma_i t_i\), \(\gamma_i \in \{0,1,2\}\). If such representation exists, it is unique, because \(B\) is a proper GAP.
  Let us denote \(c_0 := \{i \mid \gamma_i = 0\}\), \(c_1 := \{i \mid \gamma_i = 1\}\), \(c_2 := \{i \mid \gamma_i = 2\}\).
  It is clear that \(c_0 \uplus c_1 \uplus c_2 = [m]\).

  Now suppose that \(n = a + b\) for some \(a,b \in A\). But \(a = t_{0} + \sum_i \alpha_i t_i\) and \(b = t_{0} + \sum_i \beta_i t_i\), \(\alpha_i, \beta_i \in \{0,1\}\).
  We get that if \(i \in c_{0}\) or \(i \in c_{2}\) then the corresponding coefficients \(\alpha_{i}\) and \(\beta_{i}\) are uniquely determined.
  Consider \(i \in c_{1}\). Then we have two choices: either \(\alpha_{i} = 1; \beta_{i} = 0\), or \(\alpha_{i} = 0; \beta_{i} = 1\).
  Therefore, we have \(R_n(A) = 2^{|c_1(n,A)|}\).

  We have that
  \[
    E(A,A) = \sum_{n \in Z} {R_n(A)}^2 = \sum_{n \in Z} 2^{2|c_1(n,A)|}.
  \]
  Using the fact that \(|c_0(n,A)| + |c_1(n,A)| + |c_2(n,A)| = m\), we see that
  \[
    E(A,A) = \sum_{n \in Z} 2^{2|c_1(n,A)|} = \sum_{j=0}^m \binom{m}{j} 2^{m-j} 2^{2j} = \sum_{j=0}^m \binom{m}{j} 2^{m+j} \le 2^{3m}
  \]

  We can bound the Fourier bias by Theorem~\ref{thm:fourier-bias-bound}:
  \[
    \| A \|^4_{\mathcal{U}} \le \frac{1}{|Z|^3} E(A,A) - \Prob_Z(A)^4 \le \| A \|^2_{\mathcal{U}} \Prob_Z(A)
  \]
  \[
    \| A \|_{\mathcal{U}}^4 \le \frac{2^{3m}}{2^{3 \cdot 2^m}} - \frac{2^{4m}}{2^{4 \cdot 2^m}} = \frac{d^3}{2^{3d}} - \frac{d^4}{2^{4d}}
  \]
  \[
    \| A \|_{\mathcal{U}} \le \frac{d^{3/4}}{p^{3/4}}
  \]

  Finally, we have
  \[
    \varepsilon(A) = \Big( \frac{p}{d} \|A\|_{\mathcal{U}} \Big)^2 \le \frac{p^{1/2}}{d^{1/2}}.
  \]

  By substituting the definitions of \(d\) and \(m\), we prove the theorem.
\end{proof}

%

\begin{corollary}
  The depth of the circuit that computes \(U_{a}(A)\) is \(\lceil \log p - 2 \log \epsilon \rceil\).
\end{corollary}

\begin{theorem}[Circuit depth for AIKPS sequences]
  For given \(\varepsilon > 0\), let
  \begin{align*}
    R &= \{ r \mid r \text{ is prime}, (\log p)^{1+\varepsilon}/2 < r < (\log p)^{1+\varepsilon}\},\\
    S &= \{ 1, 2, \ldots, (\log p)^{1+2\varepsilon}\}, \\
    T &= \{ s \cdot r^{-1} \mid r\in R, s \in S\},
  \end{align*}
  where \(r^{-1}\) is the inverse of \(r\) modulo \(p\).

  Then the depth of the circuit that computes \(U_{a}(T)\) is less than \((1 + 2\epsilon) \log^{1+\epsilon}p\; \log\log p\).
\end{theorem}

\begin{proof}
  Let us denote the elements of \(R\) by \(r_{1}, r_{2}, \ldots\). Let \(S\cdot \{r^{-1}\}\) be a set \(\{s \cdot r^{-1} \mid s \in S\}\).

  Consider the following circuit \(\mathcal{C}_{j}\) (see Figure~\ref{fig:aikps-circuit}) with \(w=\lceil (1+2\varepsilon)\log\log p \rceil + 1\) wires.

  The circuit \(\mathcal{C}_{j}\) has depth \(\lceil (1+2\varepsilon)\log\log p \rceil + 1\) and computes the transformation \(U_{a}(S \cdot \{r_{j}^{-1}\})\). By repeating the same circuit for all \(r_{j} \in R\) we get the required circuit for \(U_{a}(T)\) (see Figure~\ref{fig:full-aikps-circuit}).

  Since \(|R| < (\log p)^{1+\varepsilon}\), we obtain that the depth of the circuit \(U_{a}(T)\) is less than
  \[
  (1 + 2\epsilon) \log^{1+\epsilon}p\; \log\log p . \qedhere
  \]

\end{proof}

\section{Numerical Experiments}\label{sec:numeric}

We conduct the following numerical experiments. We compute sets of coefficients $K$ for the automaton for the language $MOD_p$  with minimal computational error.

Finding an optimal set of coefficients  is an optimization problem with many parameters, and the running time of a brute force algorithm is large, especially with an increasing number $m$ of control qubits and large values of parameter $p$. Then, the original automaton has $2d$ states, where $d=2^m$. We observe circuits for several $m$ values and use a heuristic method for finding the optimal sets $K$ with respect to an error minimization. For this purpose, the coordinate descend method \cite{wright2015coordinate} is used.

 We find an optimal sets of coefficients for different values of $p$ and $m$ and compare computational errors of original and shallow fingerprinting algorithms for the automaton (see Figure~\ref{fig:s5}). Namely, we set $m=3,4,5$ and find sets using the coordinate descend method  for each case that minimizes $\varepsilon(K)$. Even heuristic computing, for $s>5$, takes exponentially more computational time and it is hard to implement on our devices.

One can note that difference between errors becomes bigger with increasing $m$, especially for big values $p$. The program code and numerical data are presented in a git repository \cite{paramsComputing}.

The graphics in Figure~\ref{fig:s5-prop} show a proportion of the errors of the original automaton over the errors of the shallow automaton for $m=3,4,5$ and the prime numbers until 1013.

As we see, for a number of control qubits $m=3$, the difference between the original and shallow automata errors is approximately constant. The ratio of values fluctuates between 1 and 1.2. In the case $m=4$, this ratio is approximately 1.5 for almost all observed values $p$. The ratio of errors is nearly  between 1.5 and 3, for $m=5$.

According to the results of our experiments, the circuit depth $m+1$ is enough for valid computations, while the original circuit uses $O(2^m)$ gates. Since the shallow circuit is much simpler than the original one, its implementation on real quantum machines is much easier. For instance, in such machines as IBMQ Manila or Baidu quantum computer, a ``quantum computer'' is represented by a linearly connected sequence of qubits. CX-gates can be applied only to the neighbor qubits. For such a linear structure of qubits, the shallow circuit can be implemented using $3m+3$ CX-gates. Whereas a nearest-neighbor decomposition \cite{mottonen2006decompositions} of the original circuit requires $O(d \log d)=O(m2^m)$ CX-gates.

\subsection{Numerical Experiments for Noisy Device }
The numerical experiments presented above assume that the device has no noise, i.e. all operators and measurements are always correct. At the same time, current quantum devices are in the Noisy Intermediate-Scale Quantum (NISQ) era \cite{preskill2018quantum}, which means that each operator is noisy. Therefore, we invoke a new series of computational experiments that find the set of coefficients $K$ for the automaton that recognizes the language $MOD_p$, so that we can separate members and non-members of the language on a noisy simulator of the IBMQ quantum machine. We choose $m=3$ control qubits and $p=17$. Our device uses $4$ qubits. At the same time, it is known that any deterministic device needs at least $5$ bits to recognize the $MOD_{17}$ language \cite{af98}.

For these parameters, we invoke a brute-force algorithm to find the set $K$ that satisfies two conditions:
\begin{enumerate}
    \item The probability of accepting member words should be as high as possible.
    \item The probability of accepting non-member words should be as small as possible.
\end{enumerate}
For this reason we minimize a value $diff=\Prob(6\cdot p)-max\{\Prob(r):r$ $\mod$ $p\neq 0$, $1\leq r\leq 6\cdot p+9\}$, where $\Prob(i)$ is the acceptance probability for a word of length $i$. The result set is $K=\{4, 8, 12, 6\}$. After that, we invoke the circuit for this automaton on words of length at most $7\cdot p+9=128$. We use an optimization of the circuit from \cite{ksy2024} that allows us to use the Rz operator instead of the Ry operator. It is presented in Figure \ref{fig:short-rz}. 

It is important to use this optimization because the emulator of a real IBMQ device only allows to use a limited set of possible gates. Especially, since it allows to use of the Rz operator but not the Ry operator. 

The program has been called 10000 times and we compute the number of shots that return the state $|0\rangle$ (accepting state). We use the notation $\tilde{P}(i)$ for this number, where $i$ is the length of the word. These numbers $\tilde{P}(i)$ for each length $i$ for $1\leq i\leq 129$ are presented in Section \ref{apx:noisy-devices} and in Figure \ref{fig:short-results}. We can say that it is a statistical representation of the probability $\Prob(i)$.

Analyzing the data, we can say that each member $i$ of $MOD_p$ has $\tilde{P}(i)\geq 1021$. At the same time, each non-member $i$ has $\tilde{P}(i)\leq 925$. So we can choose a threshold $\tilde{\lambda}=1000$. For any length $i$ we can say, that it is a member iff $\tilde{P}>\tilde{\lambda}$.

It is an analog of the acceptance criteria for automata using the cut-point \cite{SY14}. If we choose a cutpoint $\lambda=0.1$ and a small constant $\varepsilon=0.001$, then we can say that if a word $i$ is a member of the language, then $\Prob(i)>\lambda+\varepsilon$. If a word $i$ is a non-member, then $\Prob(i)<\lambda-\varepsilon$.

We do similar experiments for the standard circuit of fingerprinting for automata for the $MOD_p$ language. Since $m=3$, we should find $2^m=8$ parameters for angles. We cannot use the brute force algorithm for searching parameters because each invocation of the IBMQ emulate function takes time. We use the coordinate descend method for the first $6$ parameters and brute force for the last two parameters. The objective function for minimization was $diff'=\Prob(p)-\{\Prob(r):r$ $\mod$ $p\neq 0$, $1\leq r\leq 2\cdot p+10\}$. The target parameters are $K=\{14, 9, 6, 7, 10, 4, 2, 16\}$.

Then we call the circuit for this automaton on words of length at most $7\cdot p+9=128$. The program has been run 10000 times and we compute $\tilde{P}(i)$, which is the number of shots that return the state $|0\rangle$ (accepting state). These numbers $\tilde{P}(i)$ for each length $i$ for $1\leq i\leq 128$ are presented in Section \ref{apx:noisy-devices} and Figure \ref{fig:ucr-results}.

When analyzing the data, only few members like  $i=p$, $i=2p$, $i=4p$, and $i=7p$ can be separated from non-members. Other members of $MOD_p$ cannot be separated from members, for example $i=3p$, $i=5p$, and $i=6p$. In addition, we can see that the number $\tilde{P}(1)>\tilde{P}(p)$. Our algorithm cannot find a set $K$ that violates this condition. At the same time, we can ignore $i=1$ because it can be checked with additional conditions. Even in this case, the algorithm is only useful for detecting some members with a threshold of $\tilde{\lambda}=1000$, but we cannot detect other members. 

Note that $\tilde{P}(p)$,  in the case of the shallow circuit, is much higher than in the case of the standard circuit. 
For instance, $\tilde{P}(p)>5100$ for the shallow circuit, which corresponds to $\Prob(p)>0.51$; while for the standard circuit the corresponding probability would be much smaller than $0.5$ (approximately $0.1$).

Finally, we can see that on noisy devices (emulators) the shallow circuit is more efficient than the standard circuit even if for the ``ideal'' device we have an opposite situation.

\section{Conclusions}\label{sec:conclusions}

We show that generalized arithmetic progressions generate some sets of coefficients \(k_{i}\) for the quantum fingerprinting technique with provable properties. These sets have large sizes, but their depth is small and comparable to the depth of sets obtained by the probabilistic method. These sets can be used in implementations of quantum finite automata suitable for running on current quantum hardware.

We perform numerical simulations. In the case of ``ideal'' (non-noisy device), they show that the actual performance of the coefficients found by our method for quantum finite automata is not much worse than the performance of the other methods. At the same time, we show that our method is much more efficient for IBMQ noisy device emulators.

Optimizing the implementation of quantum finite automata for depth is also an open question. The lower bound on the size of \(K\) with respect to \(p\) and \(\varepsilon\) is known~\cite{Ablayev2016a}. Therefore, for given \(p\) and \(\varepsilon\), quantum finite automata cannot have less than \(O(\log p / \varepsilon)\) states. At the same time, to our knowledge, a lower bound on the circuit depth of the transition function implementation is not known. Thus, we pose an open question: Is it possible to implement a transition function with a depth less than \(O(\log p)\)? What is the lower bound?

\section{Numerical Experiments for Noisy Devices} \label{apx:noisy-devices}
The data for the shallow circuit for fingerprinting is presented in the following table.
\begin{center}
\begin{tabular}{| c| c |c| c| c |c| c| c |c| c| c |}
\hline
the length of a word & 1 & 2 & 3 & 4 & 5 & 6 & 7 & 8 & 9 & 10\\
\hline
the count & 122 & 925 & 61 & 92 & 187 & 539 & 184 & 91 & 77 & 237\\
\hline
\hline
the length of a word & 11 & 12 & 13 & 14 & 15 & 16 & \textbf{17} & 18 & 19 & 20\\
\hline
the count & 514 & 273 & 142 & 149 & 737 & 571 & \textbf{6439} & 570 & 875 & 312\\
\hline
\hline
the length of a word & 21 & 22 & 23 & 24 & 25 & 26 & 27 & 28 & 29 & 30\\
\hline
the count & 240 & 293 & 586 & 348 & 231 & 235 & 271 & 604 & 348 & 181\\
\hline
\hline
the length of a word & 31 & 32 & 33 & \textbf{34} & 35 & 36 & 37 & 38 & 39 & 40\\
\hline
the count & 268 & 626 & 793 & \textbf{3870} & 691 & 701 & 462 & 377 & 432 & 557\\
\hline
\hline
the length of a word & 41 & 42 & 43 & 44 & 45 & 46 & 47 & 48 & 49 & 50\\
\hline
the count & 471 & 354 & 358 & 365 & 603 & 402 & 412 & 360 & 737 & 845\\
\hline
\hline
the length of a word & \textbf{51} & 52 & 53 & 54 & 55 & 56 & 57 & 58 & 59 & 60\\
\hline
the count & \textbf{2608} & 756 & 613 & 453 & 427 & 564 & 665 & 466 & 440 & 396\\
\hline
\hline
the length of a word & 61 & 62 & 63 & 64 & 65 & 66 & 67 & \textbf{68} & 69 & 70\\
\hline
the count & 497 & 560 & 494 & 463 & 523 & 671 & 701 & \textbf{1981} & 758 & 754\\
\hline
\hline
the length of a word & 71 & 72 & 73 & 74 & 75 & 76 & 77 & 78 & 79 & 80\\
\hline
the count & 529 & 468 & 595 & 594 & 600 & 385 & 473 & 570 & 567 & 579\\
\hline
\hline
the length of a word & 81 & 82 & 83 & 84 & \textbf{85} & 86 & 87 & 88 & 89 & 90\\
\hline
the count & 501 & 435 & 667 & 717 & \textbf{2011} & 775 & 679 & 544 & 532 & 512\\
\hline
\hline
the length of a word & 91 & 92 & 93 & 94 & 95 & 96 & 97 & 98 & 99 & 100\\
\hline
the count & 626 & 586 & 495 & 552 & 541 & 631 & 472 & 524 & 448 & 641\\
\hline
\hline
the length of a word & 101 & \textbf{102} & 103 & 104 & 105 & 106 & 107 & 108 & 109 & 110\\
\hline
the count & 678 & \textbf{1527} & 702 & 688 & 495 & 498 & 561 & 621 & 587 & 527\\
\hline
\hline
the length of a word & 111 & 112 & 113 & 114 & 115 & 116 & 117 & 118 & \textbf{119} & 120\\
\hline
the count & 521 & 457 & 632 & 556 & 518 & 546 & 692 & 697 & \textbf{1021} & 694\\
\hline
\hline
the length of a word & 121 & 122 & 123 & 124 & 125 & 126 & 127 & 128 &   & \\
\hline
the count & 608 & 558 & 631 & 546 & 568 & 560 & 421 & 582 &   &  \\
\hline
\end{tabular}
\end{center}

The data for the standard circuit for fingerprinting is presented in the following table.
\begin{center}
\begin{tabular}{| c| c |c| c| c |c| c| c |c| c| c |}
\hline
the length of a word & 1 & 2 & 3 & 4 & 5 & 6 & 7 & 8 & 9 & 10\\
\hline
the count & 4042 & 82 & 654 & 388 & 217 & 463 & 280 & 210 & 364 & 318\\
\hline
\hline
the length of a word & 11 & 12 & 13 & 14 & 15 & 16 & \textbf{17} & 18 & 19 & 20\\
\hline
the count & 596 & 459 & 311 & 419 & 465 & 558 & \textbf{1005} & 589 & 324 & 547\\
\hline
\hline
the length of a word & 21 & 22 & 23 & 24 & 25 & 26 & 27 & 28 & 29 & 30\\
\hline
the count & 561 & 366 & 467 & 446 & 528 & 649 & 823 & 459 & 622 & 659\\
\hline
\hline
the length of a word & 31 & 32 & 33 & \textbf{34} & 35 & 36 & 37 & 38 & 39 & 40\\
\hline
the count & 372 & 441 & 735 & \textbf{1053} & 662 & 410 & 450 & 637 & 618 & 507\\
\hline
\hline
the length of a word & 41 & 42 & 43 & 44 & 45 & 46 & 47 & 48 & 49 & 50\\
\hline
the count & 614 & 623 & 524 & 517 & 545 & 534 & 608 & 626 & 476 & 634\\
\hline
\hline
the length of a word & \textbf{51} & 52 & 53 & 54 & 55 & 56 & 57 & 58 & 59 & 60\\
\hline
the count & \textbf{746} & 712 & 604 & 500 & 531 & 513 & 624 & 571 & 602 & 571\\
\hline
\hline
the length of a word & 61 & 62 & 63 & 64 & 65 & 66 & 67 & \textbf{68} & 69 & 70\\
\hline
the count & 475 & 635 & 558 & 569 & 654 & 549 & 767 & \textbf{1160} & 727 & 603\\
\hline
\hline
the length of a word & 71 & 72 & 73 & 74 & 75 & 76 & 77 & 78 & 79 & 80\\
\hline
the count & 631 & 607 & 565 & 617 & 576 & 564 & 542 & 609 & 625 & 620\\
\hline
\hline
the length of a word & 81 & 82 & 83 & 84 & \textbf{85} & 86 & 87 & 88 & 89 & 90\\
\hline
the count & 573 & 610 & 622 & 617 & \textbf{712} & 630 & 629 & 591 & 628 & 566\\
\hline
\hline
the length of a word & 91 & 92 & 93 & 94 & 95 & 96 & 97 & 98 & 99 & 100\\
\hline
the count & 647 & 623 & 579 & 608 & 654 & 649 & 629 & 611 & 625 & 593\\
\hline
\hline
the length of a word & 101 & \textbf{102} & 103 & 104 & 105 & 106 & 107 & 108 & 109 & 110\\
\hline
the count & 654 & \textbf{682} & 628 & 592 & 604 & 628 & 639 & 589 & 614 & 652\\
\hline
\hline
the length of a word & 111 & 112 & 113 & 114 & 115 & 116 & 117 & 118 & \textbf{119} & 120\\
\hline
the count & 625 & 637 & 599 & 606 & 612 & 684 & 648 & 627 & \textbf{828} & 655\\
\hline
\hline
the length of a word & 121 & 122 & 123 & 124 & 125 & 126 & 127 & 128&   &  \\
\hline
the count & 603 & 662 & 635 & 637 & 591 & 651 & 613 & 624 &   &  \\
\hline
\end{tabular}
\end{center}

\section*{Figures}


\begin{figure}[ht!]
  \begin{quantikz}
    \lstick{$\ket{0}$} & \gate{H} & \ctrl{5}   & \ctrl{5}   & \ctrl{5}   & \ctrl{5}   & \qw\ldots & \octrl{5}  & \qw \\
    \lstick{$\ket{0}$} & \gate{H} & \ctrl{4}   & \ctrl{4}   & \ctrl{4}   & \ctrl{4}   & \qw\ldots & \octrl{4}  & \qw \\
    \wave&&&&&&&&& \\
    \lstick{$\ket{0}$} & \gate{H} & \ctrl{2}   & \ctrl{2}   & \octrl{2}  & \octrl{2}  & \qw\ldots & \octrl{2}  & \qw \\
    \lstick{$\ket{0}$} & \gate{H} & \ctrl{1}   & \octrl{1}  & \ctrl{1}   & \octrl{1}  & \qw\ldots & \octrl{1}  & \qw \\
    \lstick{$\ket{0}$} & \qw      & \gate{U_1} & \gate{U_2} & \gate{U_3} & \gate{U_4} & \qw\ldots & \gate{U_t} & \qw
  \end{quantikz}
  \caption{Deep fingerprinting circuit example. Gate \(U_{j}\) is a rotation \(R_{y}(4\pi k_{j}x / p)\). Controls in controlled gates run over all binary strings of length \(s\)} \label{fig:deep-circuit}
\end{figure}
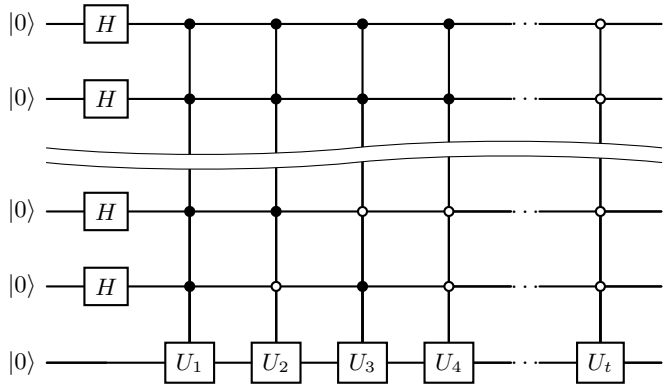

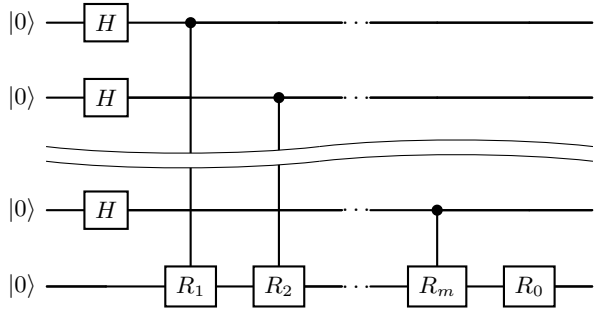
\begin{figure}[ht!]
  \begin{quantikz}
  \lstick{$\ket{0}$} & \gate{H} & \ctrl{4} & \qw      & \qw\ldots & \qw      & \qw & \qw\\
  \lstick{$\ket{0}$} & \gate{H} & \qw      & \ctrl{3} & \qw\ldots & \qw      & \qw & \qw\\
  \wave&&&&&&& \\
  \lstick{$\ket{0}$} & \gate{H} & \qw      & \qw      & \qw\ldots & \ctrl{1} & \qw & \qw\\
  \lstick{$\ket{0}$} &  \qw     &\gate{R_1}&\gate{R_2}& \qw\ldots &\gate{R_m}& \gate{R_0} & \qw
  \end{quantikz}
  \caption{Shallow fingerprinting circuit example. Gate \(R_{j}\) is a rotation \(R_{y}(4\pi t_{j}x / p)\)} \label{fig:shallow-circuit}
\end{figure}

 \begin{figure}[ht!]
    \begin{quantikz}
      \lstick{$\ket{0}$} & \gate{H} & \ctrl{4} & \qw      & \qw\ldots & \qw      & \qw & \qw\\
      \lstick{$\ket{0}$} & \gate{H} & \qw      & \ctrl{3} & \qw\ldots & \qw      & \qw & \qw\\
      \wave&&&&&&& \\
      \lstick{$\ket{0}$} & \gate{H} & \qw      & \qw      & \qw\ldots & \ctrl{1} & \qw & \qw\\
      \lstick{$\ket{0}$} &  \qw     &\gate{R_{j,1}}&\gate{R_{j,2}}& \qw\ldots &\gate{R_{j,w-1}}& \gate{R_{j}} & \qw
    \end{quantikz}
    \caption{Circuit \(\mathcal{C}_j\) for AIKPS subsequence. Gate \(R_{j}\) is a rotation \(R_{y}(4\pi (r_{j}^{-1}) / p)\). Gate \(R_{j,k}\) is a rotation \(R_{y}(2^{k-1} \cdot 4\pi (r_{j}^{-1}) / p)\)} \label{fig:aikps-circuit}
  \end{figure}
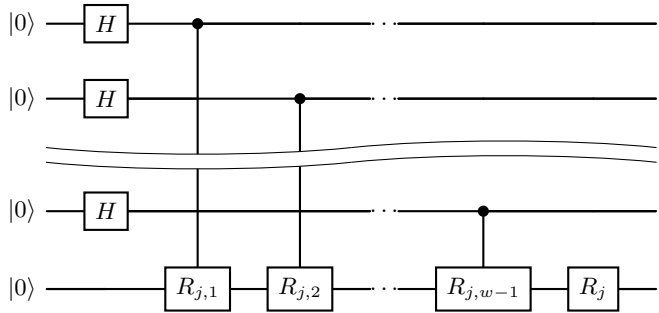

  \begin{figure}[h!]
%
%
%
  \includegraphics[width=0.8\textwidth]{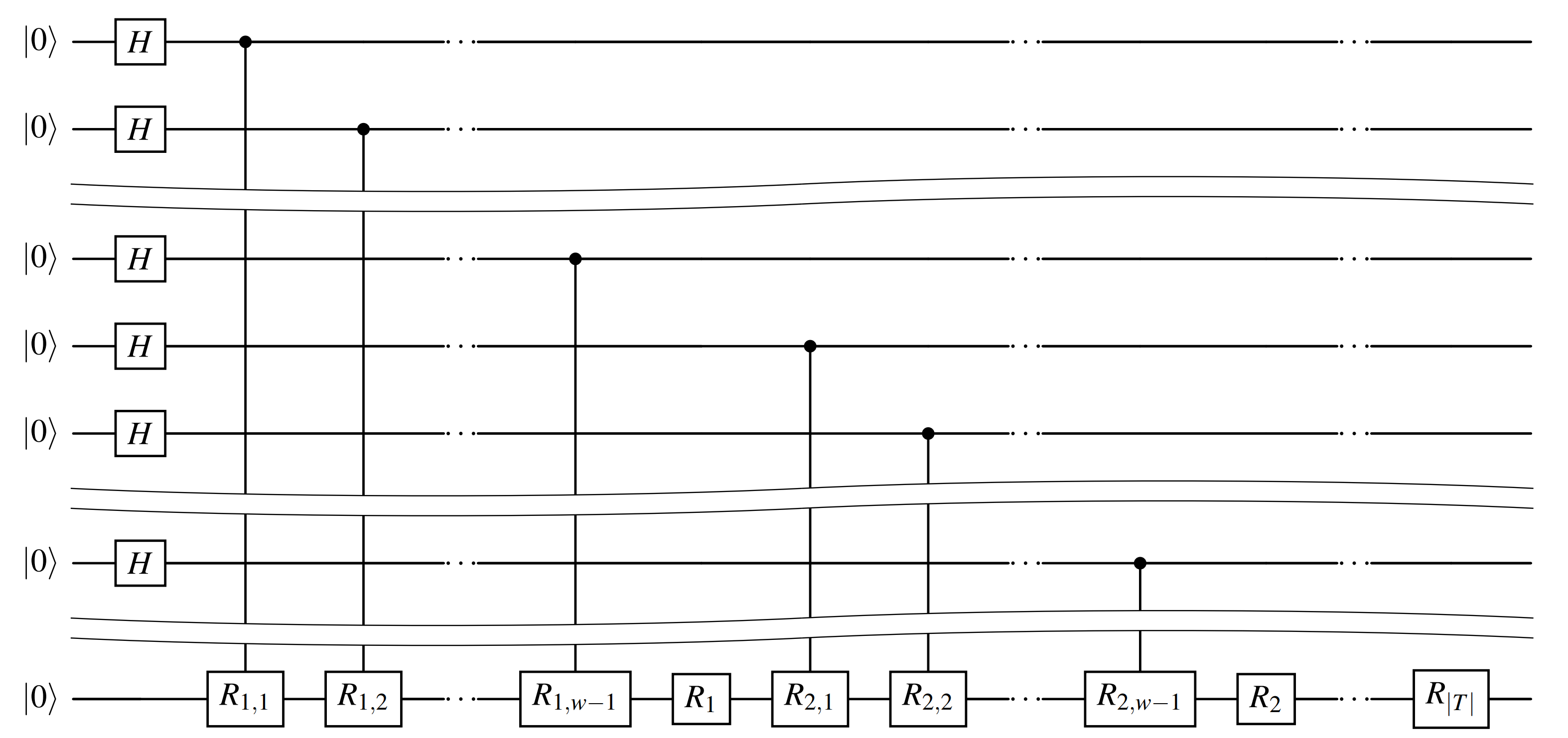}
  \caption{Circuit for \(U_{a}(T)\). Gate \(R_{j}\) is a rotation \(R_{y}(4\pi (r_{j}^{-1}) / p)\). Gate \(R_{j,k}\) is a rotation \(R_{y}(2^{k-1} \cdot 4\pi (r_{j}^{-1}) / p)\)} \label{fig:full-aikps-circuit}
  \end{figure}

\begin{figure}[ht!]
\caption{Computational errors for $m=3,4,5$ of original and shallow automata}
\label{fig:s5}
\centering
\includegraphics[width=0.75\textwidth]{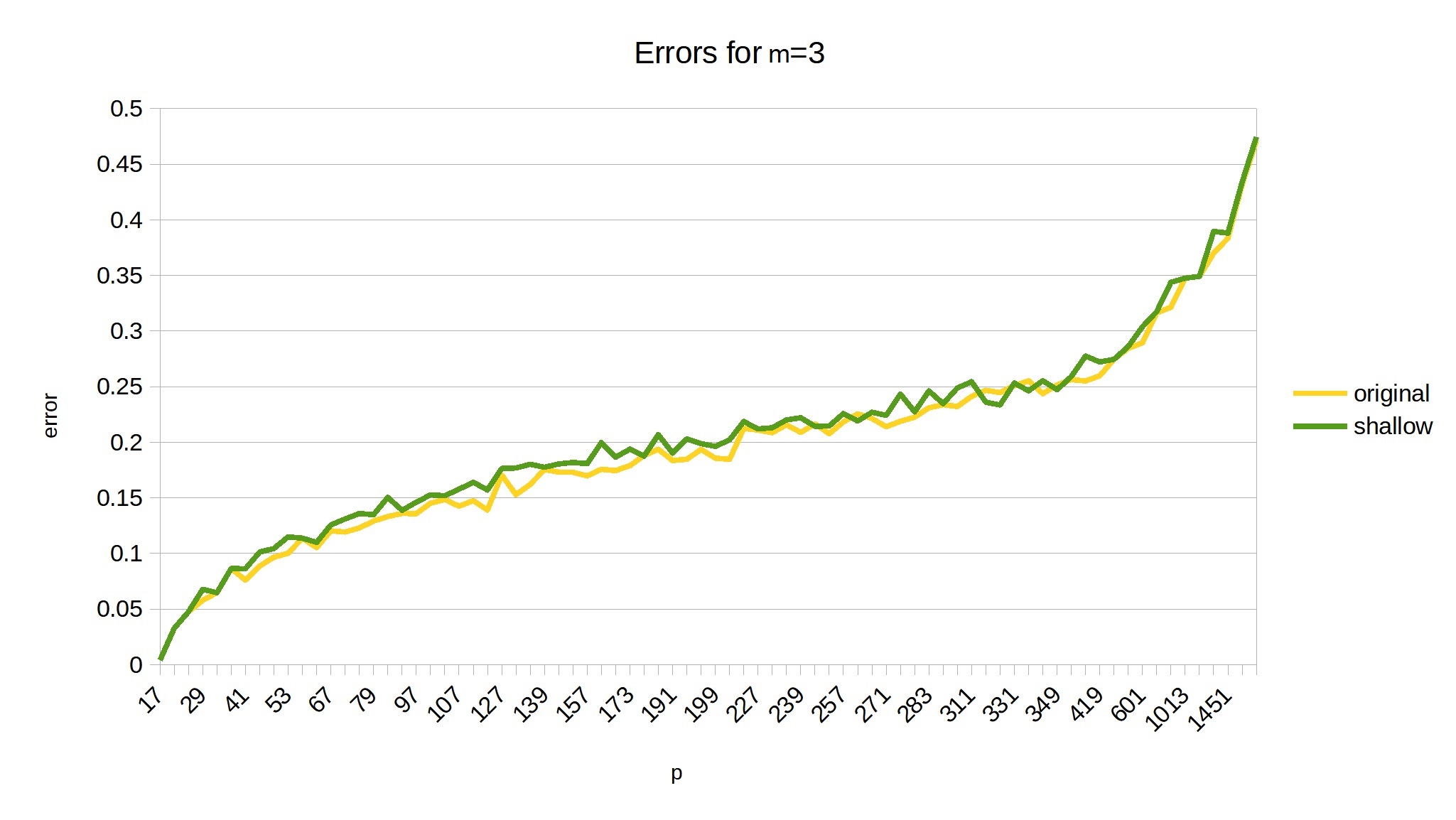}
\includegraphics[width=0.75\textwidth]{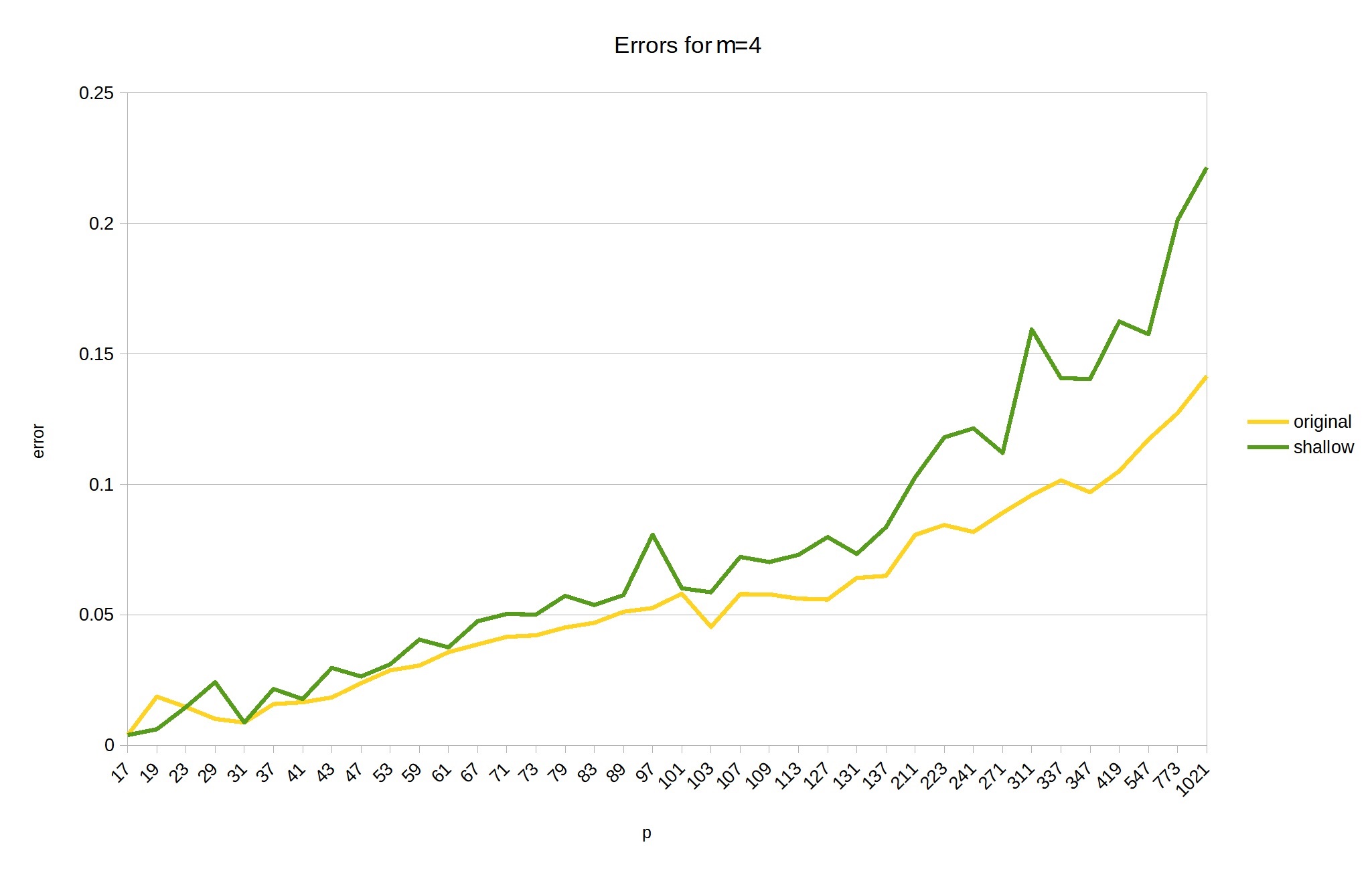}
\includegraphics[width=0.75\textwidth]{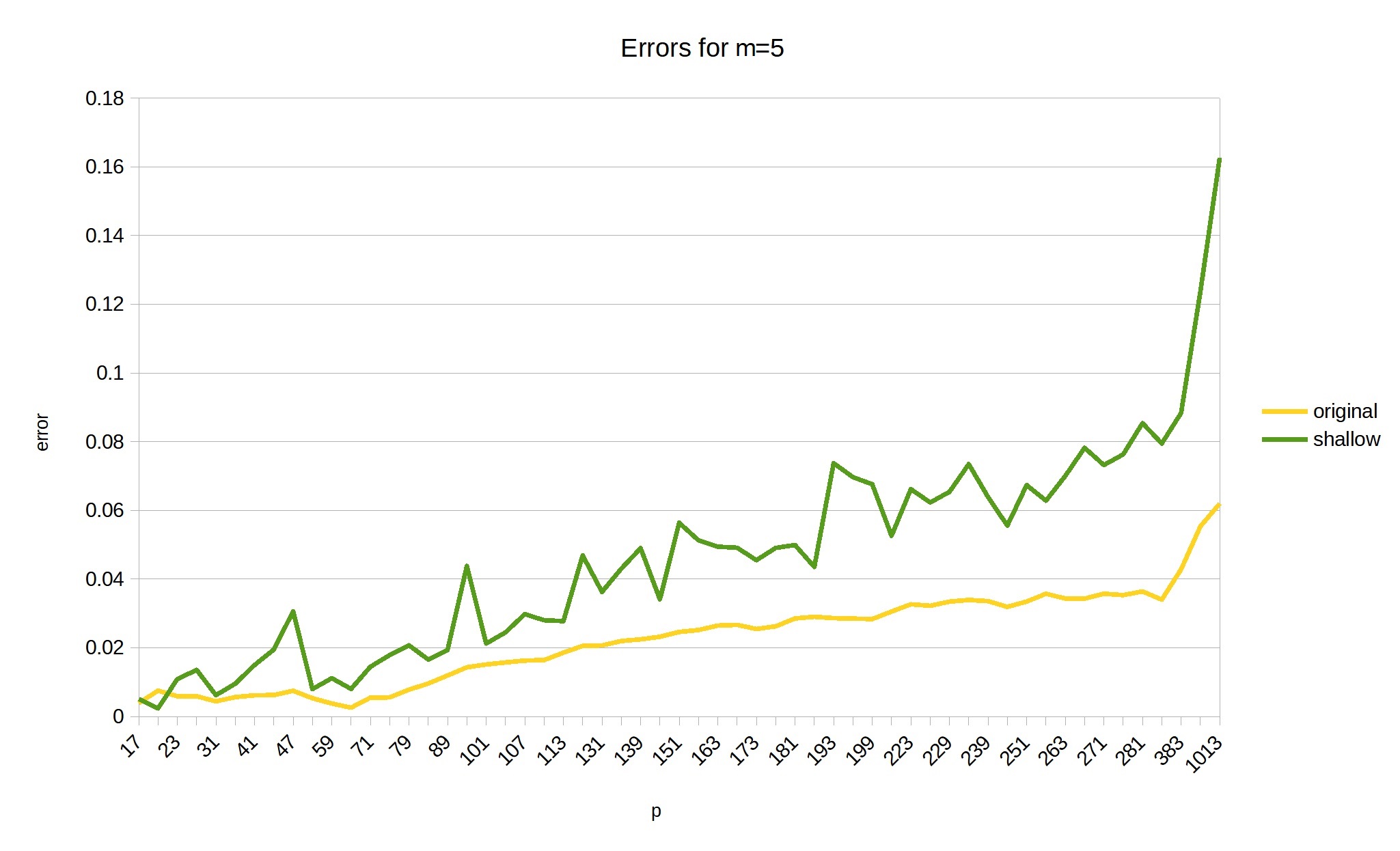}
\end{figure}  

\begin{figure}[ht!]
\caption{Proportions of the shallow automaton errors over the original automaton errors for $m=3,4,5$ and different values of $p$}
\label{fig:s5-prop}
\centering
\includegraphics[width=0.75\textwidth]{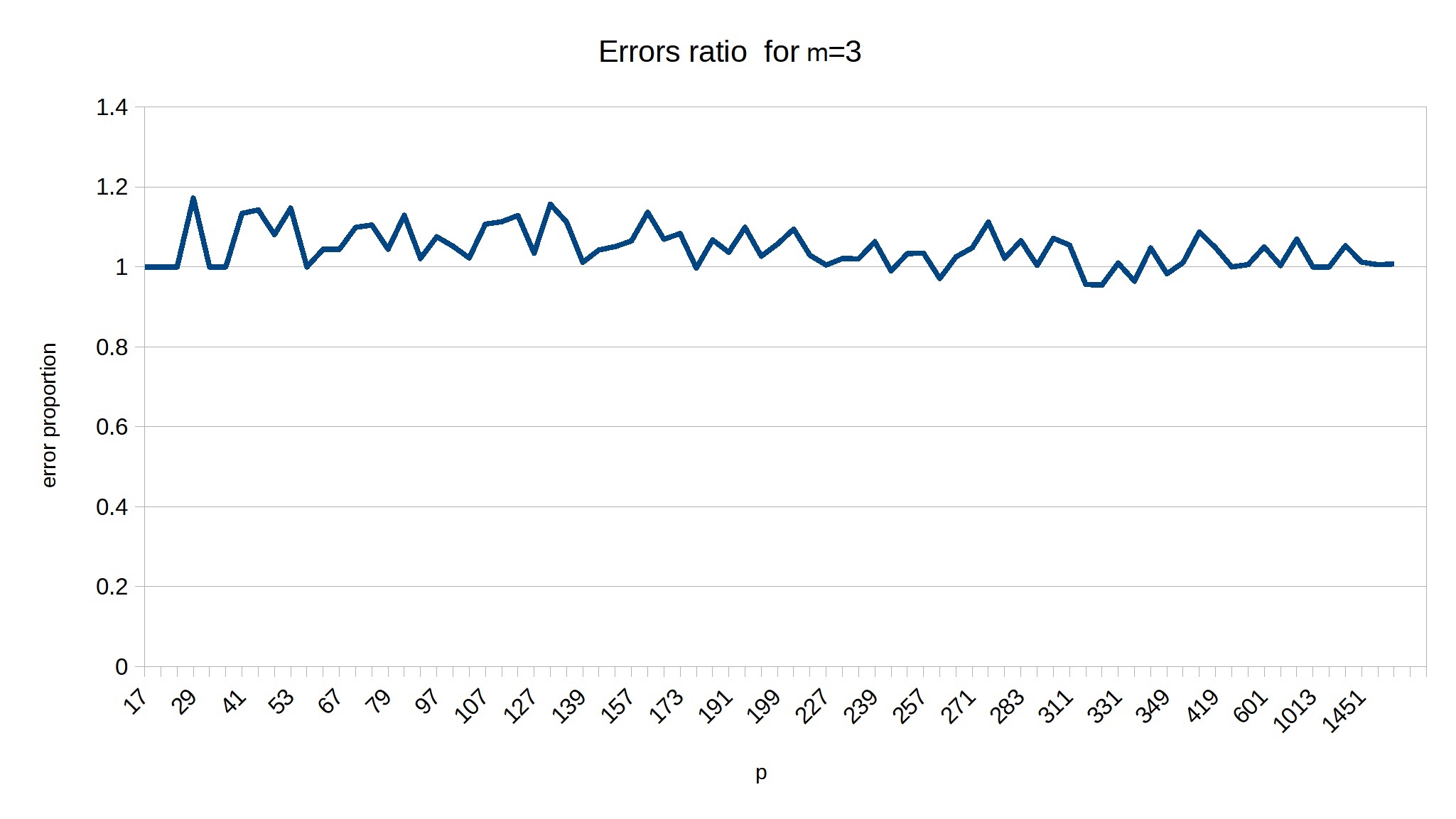}
\includegraphics[width=0.75\textwidth]{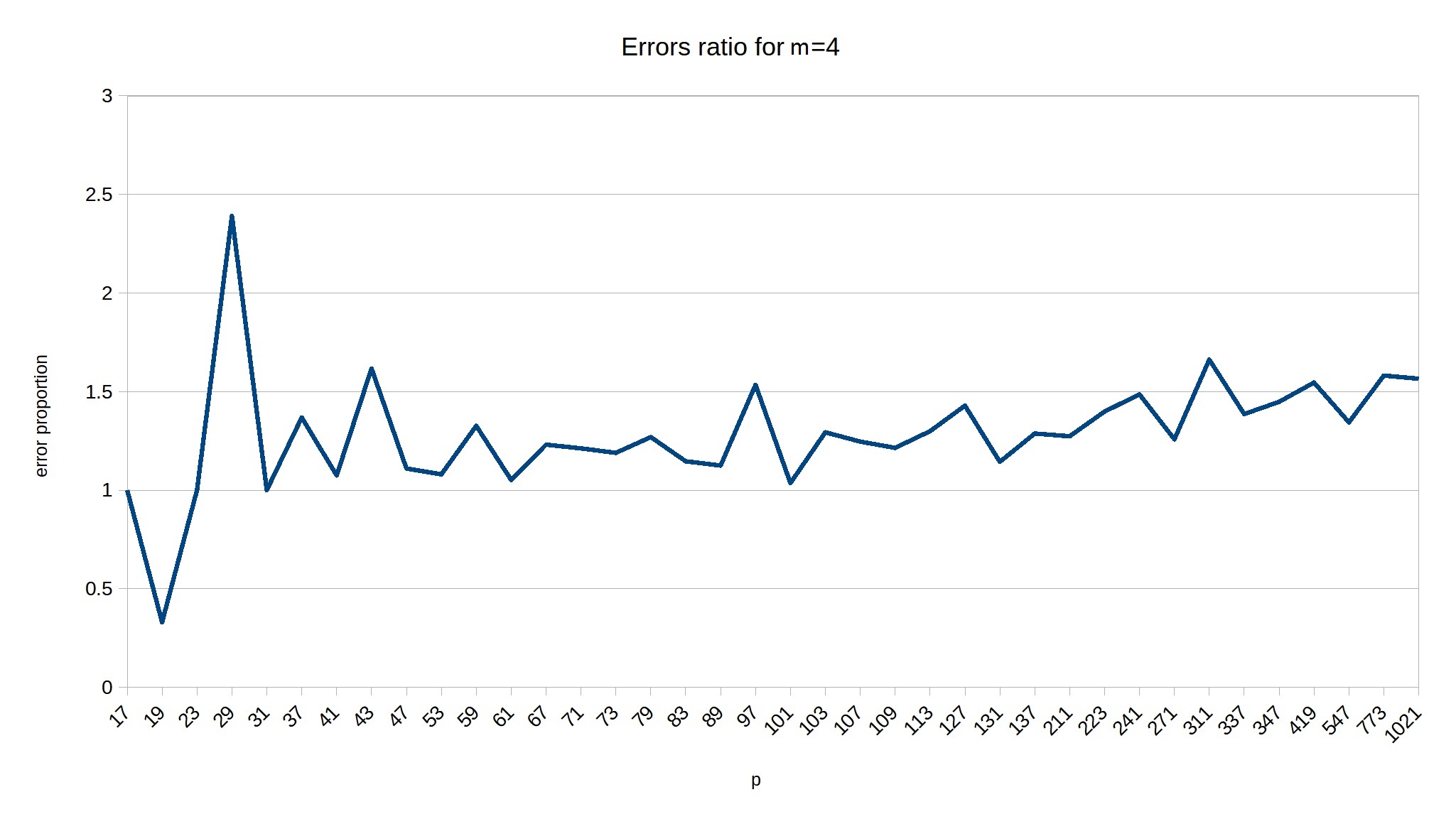}
\includegraphics[width=0.75\textwidth]{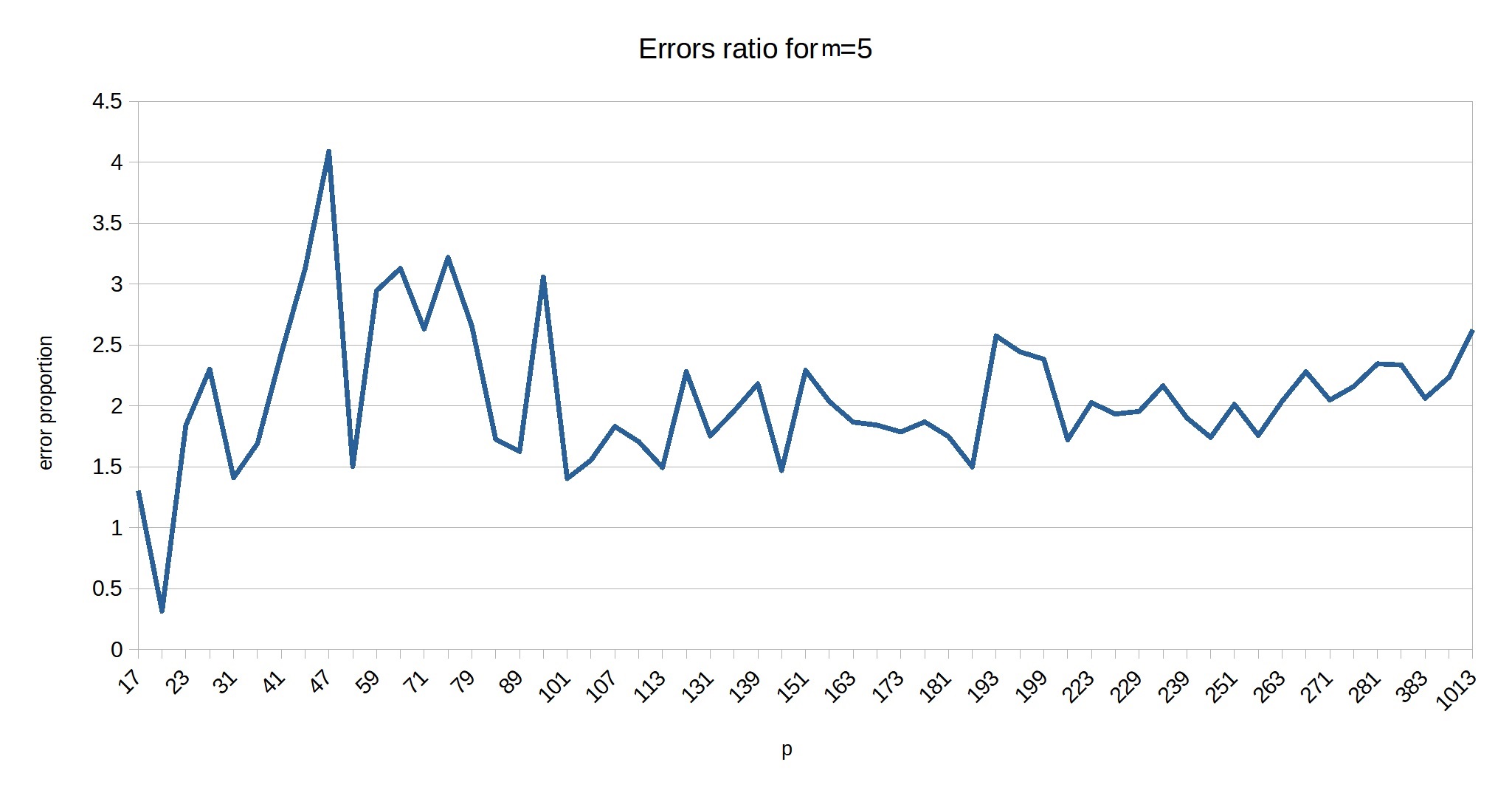}
\end{figure}
 
\begin{figure}
\includegraphics[width=0.8\textwidth]{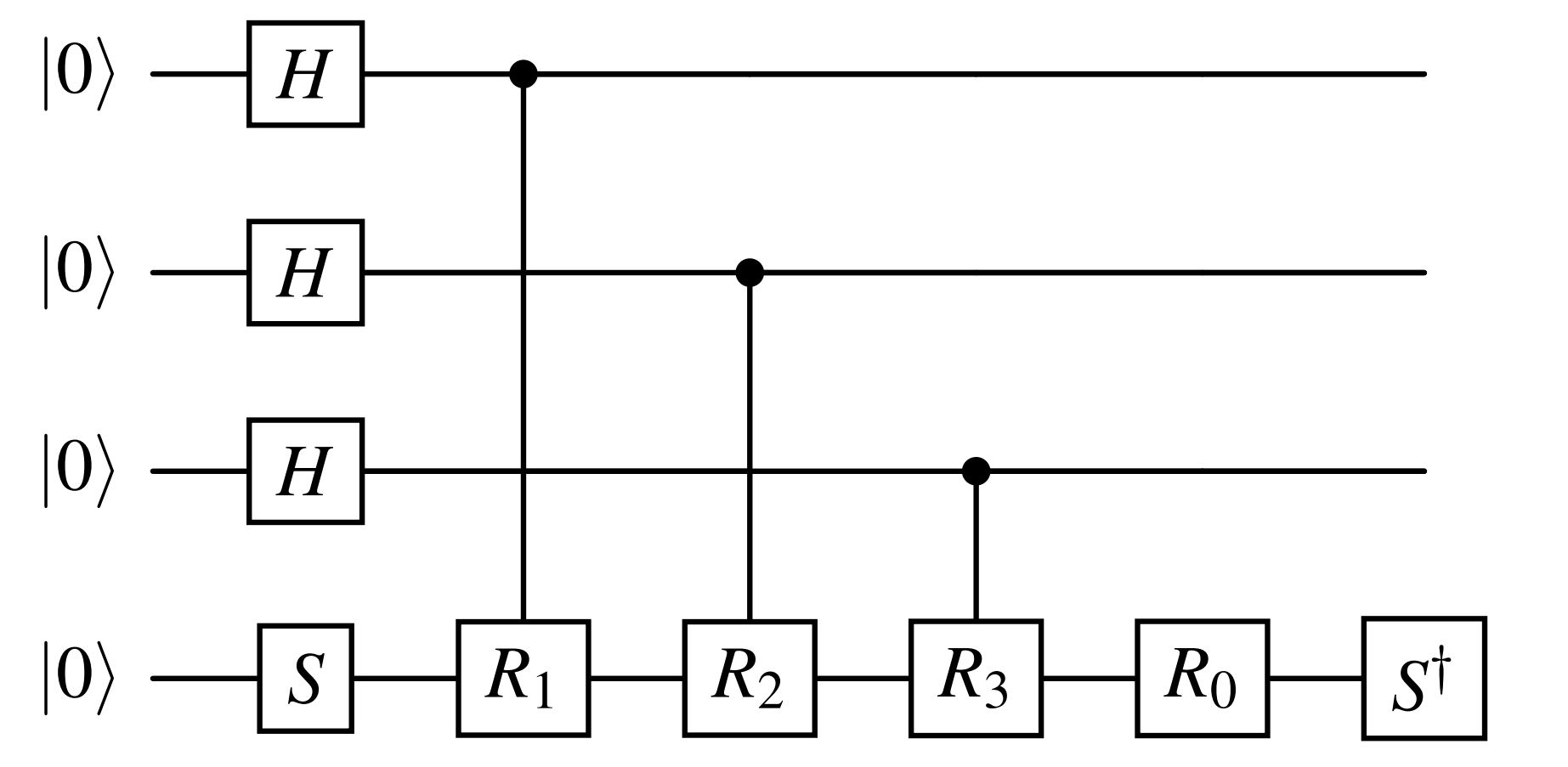}
  \caption{Shallow circuit that uses Rz operators} \label{fig:short-rz}
\end{figure}

\begin{figure}[ht!]
\caption{Computational errors for $m=3$ of shallow automata for noisy device}
\label{fig:short-results}
\centering
\includegraphics[width=0.8\textwidth]{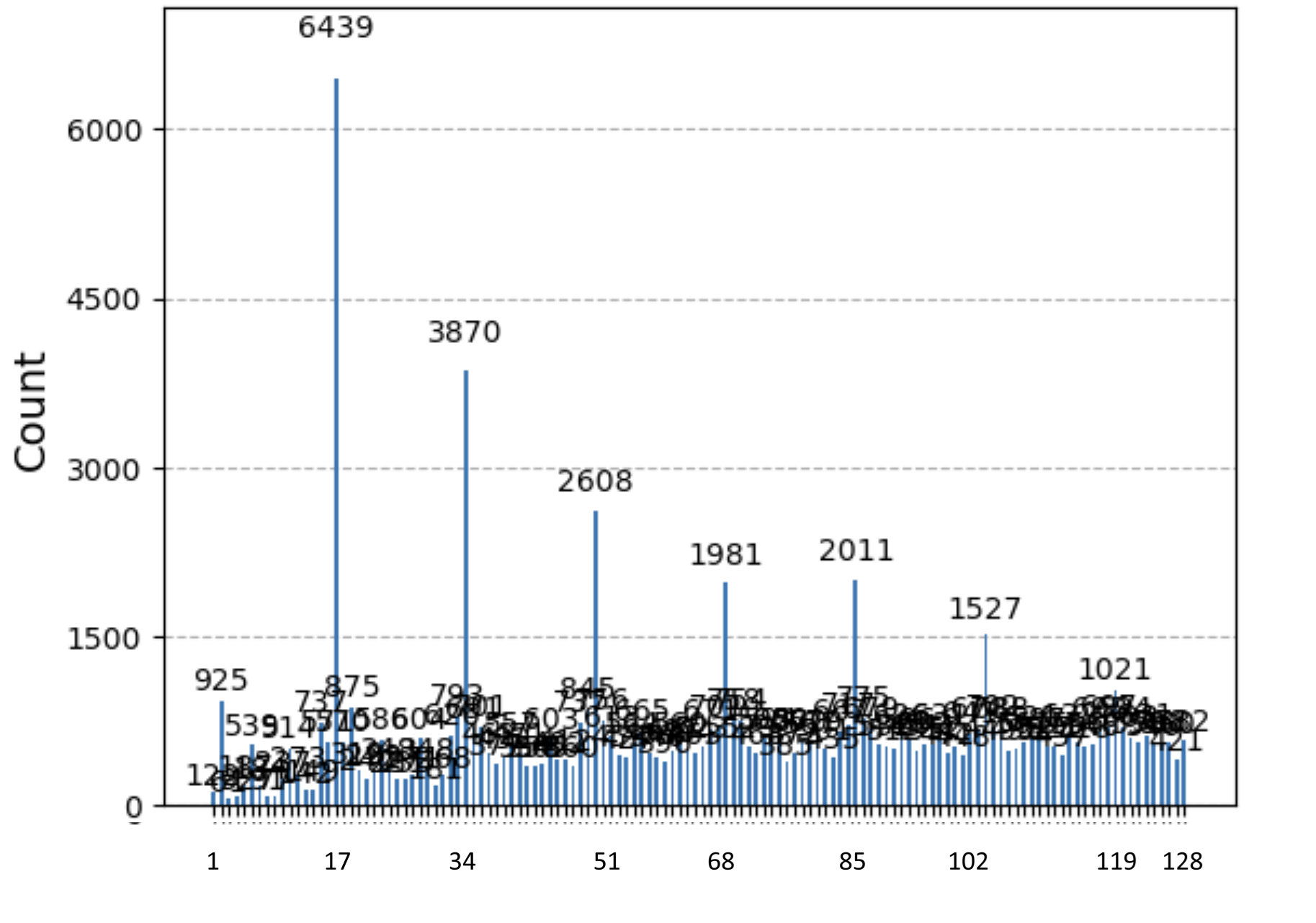}
\end{figure}

\begin{figure}[ht!]
\caption{Computational errors for $m=3$ of standard fingerprinting scheme automata for noisy device}
\label{fig:ucr-results}
\centering
\includegraphics[width=0.8\textwidth]{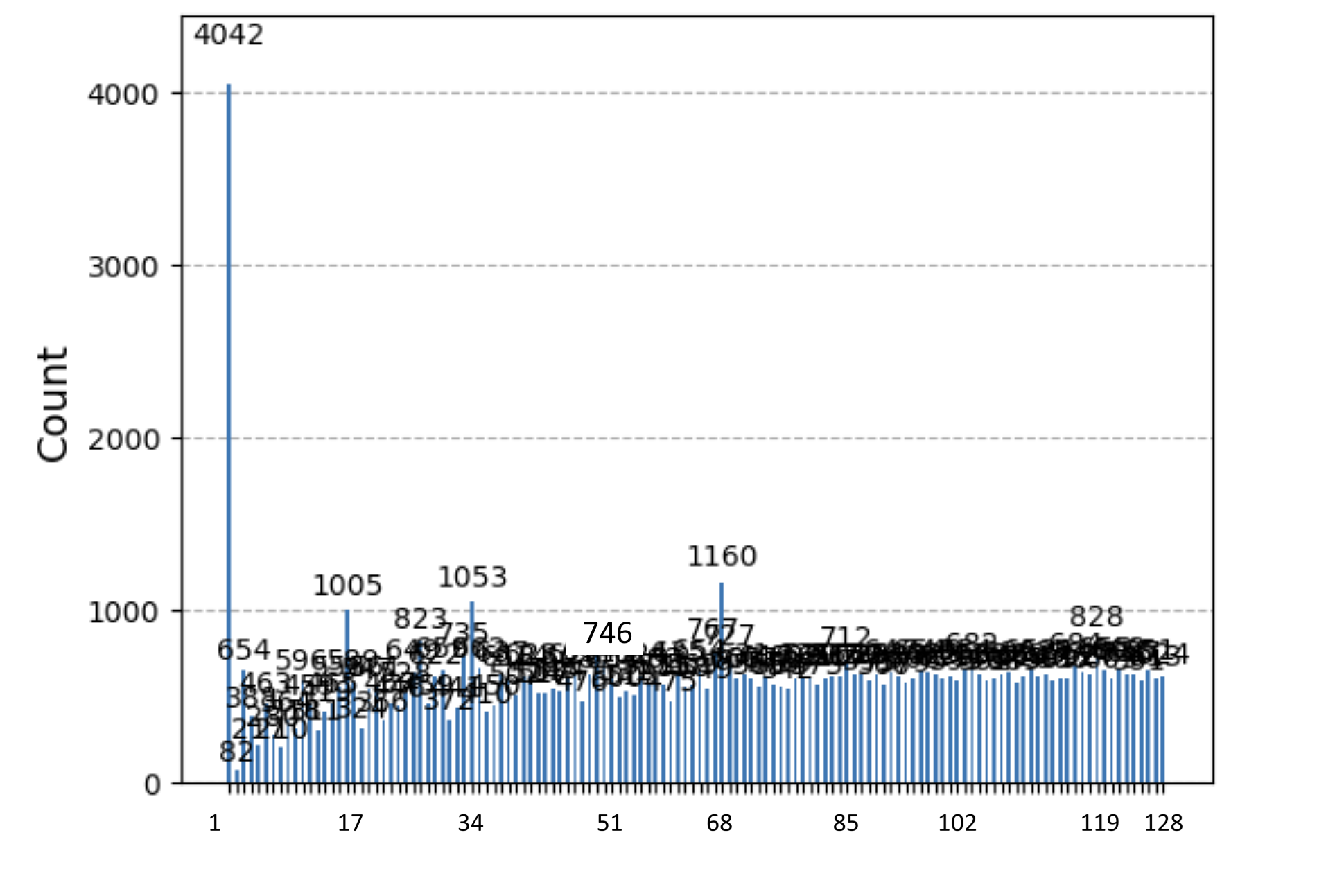}
\end{figure}

\newpage

\newpage

%
%
%
\bibliographystyle{splncs04}
\bibliography{test}


%
%
%
\end{document}